\documentclass[letterpaper]{article} 
\usepackage[authorversion]{aaai24}  
\usepackage{times}  
\usepackage{helvet}  
\usepackage{courier}  
\usepackage[hyphens]{url}  
\usepackage{graphicx} 
\usepackage{natbib}  
\usepackage{caption} 
\usepackage{amsmath,amssymb}
\usepackage{booktabs}
\usepackage{color}

\usepackage{longtable}%
\usepackage{multicol}%
\usepackage{fancybox}
\usepackage[subnum]{cases}

\usepackage[noend]{algpseudocode}

\usepackage{algorithm}

\def\QQ{\mathbb Q}
\def\RR{\mathbb R}

\def\cA{\mathcal A}

\def\cF{\mathcal F}

\def\cX{\mathcal X}

\newcommand{\bone}{\ensuremath{{\bf 1}}}
\newcommand{\bo}{\ensuremath{{\bf 0}}} 

\newcommand{\bx}{\ensuremath{{\bf x}}}
\newcommand{\by}{\ensuremath{{\bf y}}}
\newcommand{\bz}{\ensuremath{{\bf z}}}
\newcommand{\ba}{\ensuremath{{\bf a}}}

\newcommand{\bc}{\ensuremath{{\bf c}}}
\newcommand{\bd}{\ensuremath{{\bf d}}}
\newcommand{\bq}{\ensuremath{{\bf q}}}
\newcommand{\bp}{\ensuremath{{\bf p}}}

\newcommand{\opt}{\ensuremath{\mathrm{OPT}}}
\newcommand{\p}{\ensuremath{\mathrm{P}}}
\newcommand{\np}{\ensuremath{\mathrm{NP}}}

\newcommand{\ub}{\ensuremath{\mathsf{UB}}}
\newcommand{\lb}{\ensuremath{\mathsf{LB}}}

\newcommand{\poly}{\ensuremath{\mathsf{poly}}}
\newcommand{\polylog}{\ensuremath{\mathsf{polylog}}}
\newcommand{\argmax}{\ensuremath{\mathsf{argmax}}}

\newcommand{\mCKP}{\ensuremath{\mathsf{MinCKP}}}
\newcommand{\mSUB}{\ensuremath{\mathsf{MaxSUB}}}

\newcommand{\raf}[1]{(\ref{#1})}

\newcommand{\lp}{\ensuremath{{\mathsf {LP}}}}
\newcommand{\nlp}{\ensuremath{\mathsf{NLP}}}

\newtheorem{theorem}{Theorem}

\newtheorem{lemma}{Lemma}
\newtheorem{corollary}{Corollary}
\newtheorem{observation}{Observation}

\usepackage{pifont}
\newcommand\qedblob{\ding{113}}
\def\literalqed{{\ \nolinebreak\hfill\mbox{\qedblob\quad}}}

\newenvironment{proof}{\noindent{\bf Proof.}\hspace*{1em}}{\literalqed\medskip}

\usepackage{newfloat}
\usepackage{listings}
\DeclareCaptionStyle{ruled}{labelfont=normalfont,labelsep=colon,strut=off} 
\floatstyle{ruled}
\newfloat{listing}{tb}{lst}{}
\floatname{listing}{Listing}
\title{Improved Approximation Guarantees for Power Scheduling Problems With Sum-of-Squares Constraints}
\author {
    Trung Thanh Nguyen\textsuperscript{\rm 1},
    Khaled Elbassioni\textsuperscript{\rm 2},
    Areg Karapetyan\textsuperscript{\rm 3},
    Majid Khonji\textsuperscript{\rm 2}
}
\affiliations {
    \textsuperscript{\rm 1}ORLab,  Faculty of Computer Science, Phenikaa University, 12116 Hanoi, Vietnam\\
    \textsuperscript{\rm 2}Department of Electrical Engineering and Computer Science, Khalifa University, UAE\\
    \textsuperscript{\rm 3}Division of Engineering, New York University Abu Dhabi, UAE\\
thanh.nguyentrung@phenikaa-uni.edu.vn, khaled.elbassioni@ku.ac.ae, areg.karapetyan@nyu.edu,
majid.khonji@ku.ac.ae
}

\begin{document}

\maketitle

\begin{abstract}
We study a class of combinatorial scheduling problems characterized by a particular type of constraint often associated with electrical power or gas energy. This constraint appears in several practical applications and is expressed as a  sum of squares of linear functions. Its nonlinear nature adds complexity to the scheduling problem, rendering it notably challenging, even in the case of a linear objective. In fact, exact polynomial time algorithms are unlikely to exist, and thus, prior works have focused on designing approximation algorithms with polynomial running time and provable guarantees on the solution quality. In an effort to advance this line of research, we present novel approximation algorithms yielding significant improvements over the existing state-of-the-art results for these problems.

\end{abstract}

\section{Introduction}
Management of power systems involves operational routines concerned with the optimal allocation/dispatch of limited resources (power/generation) under various types of objectives and constraints, including total fuel cost, total profit, emissions, active power loss, and voltage magnitude deviation, to name a few. These problems are conventionally modeled as quadratic programs (possibly with discrete variables), which are in general hard to solve exactly or approximately.  
In this paper, we focus on a class of scheduling problems where the energy constraints are modeled as a sum of squares of linear functions. Such constraints have been studied in a number of works, including electric vehicle charging in smart grids~\cite{YuC13, KarapetyanKCEZ18, ElbassioniKN19, KhonjiKEC19}; welfare maximization in gas supply networks \cite{KlimmPRS22}, and scheduling of tasks on machines  \cite{tang12, BansalKP04}.  

In \cite{YuC13}, the authors investigated the so-called {\em Complex-Demand Knapsack Problem} (CKP), formulated to optimize the distribution of power among consumers within Alternating Current (AC) electrical systems. In CKP, power (representing user demand) is expressed in complex numbers \cite{Grainger}, where the real part corresponds to active power, the imaginary part stands for reactive power, and the magnitude of the complex number measures the apparent power. Each customer possesses a specific power demand and receives a certain utility if their demand is fully met. However, due to physical limitations on the magnitude of the total power supply available in the system, only a subset of users can be selected for power allocation. Mathematically, the power constraint is formulated as a quadratic function, which is the sum of two squares of linear functions.

The sum-of-squares constraint has also appeared in the context of gas supply networks \cite{KlimmPRS22}. As noted in \cite{KlimmPRS22}, a gas pipeline network can be modeled as a directed path graph with different entry and exit nodes. Transportation requests are characterized by their entry and exit nodes and the amounts of gas to be
transported. The gas flow on an edge from node $u$ to a node $v$ is modeled by the Weymouth equations~\cite{W12}, describing the difference of the squared pressures between these nodes. For the operation of gas networks, one needs to choose a subset of transportation requests to be served such that the maximum difference of the squared pressures in the network does not exceed a given threshold. This pressure constraint can then be formulated as a sum of squares of linear functions.

Mathematically, the above two problems can be cast as a binary program that seeks to maximize a linear objective function subject to a sum-of-squares constraint. This problem covers the classical Knapsack problem as a special case and is thus $\np$-hard. The first algorithmic result was given in \cite{Woeginger00}, where it was proved that the problem is \textit{strongly} $\np$-hard, and thus does not admit a fully polynomial time approximation scheme (FPTAS). Yu and Chau \cite{YuC13} developed a $(0.5-\epsilon)$-approximation algorithm, which was subsequently improved to a PTAS in \cite{ChauEK16}, the best possible result one can hope for the problem. In \cite{ElbassioniN17} the authors further extended this result to the case where the constraint contains a constant number of squares and devised a ($\frac{1}{e}-\epsilon$)-approximation algorithm for variant with a submodular objective. Recently, \cite{KlimmPRS22} studied the version with an unbounded number of squares and presented a constant $0.618$-approximation algorithm for a linear objective and ruled out the existence of an approximation factor of $0.99$. These results rely on the so-called \textit{pipage rounding} technique. One interesting question raised here is whether this rounding technique can be extended to nonlinear objectives such as quadratic and submodular functions, which are frequently encountered in real-world economic contexts.

  
The minimization version of CKP ($\mCKP$) with a linear objective and a sum-of-two-squares constraint has been explored in \cite{ElbassioniKN19}. As noted therein, this version constitutes a special case of the Economic Dispatch problem (a.k.a, Unit Commitment, Generation Dispatch Control), which is principal to power systems engineering and is responsible for minimizing generation costs while maintaining the balance between the net supply and demand~\cite{allen12}. Unlike the maximization version, however, $\mCKP$'s natural relaxation is non-convex and hence harder to approximate. As of now, the best-known approximation was the quasi-PTAS (QPTAS) attained in \cite{ElbassioniKN19} via a geometric approach, whereas the existence of a PTAS remained an open question.

Taking a step forward, we present two novel approximations for $\mCKP$ and a generalized variant of CKP. More concretely, the contributions of this study are as follows: 

\begin{itemize}
    \item First, we provide a PTAS for $\mCKP$, which is a notable improvement upon the state-of-the-art result of a QPTAS provided in~\cite{ElbassioniKN19}. At the heart of the approximation scheme is a polynomial time algorithm for solving the (non-convex) relaxation problem to optimality, by exploiting the special structure of optimal solutions. Note that a PTAS \textit{is the best possible result} one can hope for, unless $\p=\np$, as one can use a technique similar to that in \cite{Woeginger00}\footnote{In \cite{Woeginger00} the author provided a reduction from a variant of \textsc{Partition} to disprove the existence of an FPTAS for the CKP problem. With some minor modifications, this reduction can also be established for the minimization version of CKP.} to rule out the non-existence of an FPTAS. 

    \item Next, we show that any $\alpha$-approximation algorithm for the relaxation problem maximizing a (non-monotone) submodular quadratic function over a general (non-negative) convex quadratic constraint can lead to an $(\frac{\alpha\phi^2}{2})$-approximation algorithm for the discrete case, where $\phi=\frac{2}{1+\sqrt{5}}$ is the reverse golden ratio. The idea is to generalize the technique of~\cite{KlimmPfetsch} developed for a similar problem with a linear objective: given a fractional solution, round it in such a way that the resulting discrete solution remains feasible and the objective value does not decrease. Since in the studied problem both the objective and constraint are quadratic, we have to develop a more sophisticated rounding technique than that derived in~\cite{KlimmPRS22} for the linear objective case. 
\end{itemize}

\paragraph*{Remark} A common technique for designing approximation algorithms is to first solve the relaxation to obtain an optimal (fractional) solution, then to round it to an integer one, with a small loss in the objective value. It should be noted that, however, $\mCKP$'s straightforward relaxation is {\it non-convex}, and thus, finding an efficiently computable lower/upper bound on the objective is by itself a major challenge. 

\section{Problem Formulations and Notation}

In this section, we formalize the mathematical models of the two power scheduling problems under study. The first one, termed $\mCKP$, is the minimization version of CKP \cite{YuC13} and is defined as
\begin{align}
	(\mCKP)~\min_{\bx\in\{0,1\}^n} \quad& \displaystyle \bc^\top\bx \nonumber\\
	\text{\em s.t.}\quad &   \displaystyle (\bp^\top\bx)^2 +(\bq^\top\bx)^2 \ge C,\label{mqclp-1}
\end{align}
where $\bp,\bq\in\mathbb{Q}_+^{n}$, $\bc\in\mathbb{Q}^{n}$ and $C>0$. We may assume,  w.l.o.g.,  that $\bc>0$ (if $c_i\le 0$ for some $i\in[n]$, the corresponding variable $x_i$ can be set to $1$). Also, it is assumed that $\bp,\bq$ are non-negative, otherwise the problem doesn't admit an approximation algorithm with finite ratio~\cite{ChauEK16}. In \mCKP, the scalar $C$  denotes the aggregate apparent power demand to be met, while the column vectors $\bp^\top$ and $\bq^\top$ capture the active and reactive power outputs of the contracted generation units. Accordingly, the vector $\bc$ represents the costs associated with dispatching these units. These costs may quantify the expenses carried by the load-serving entity (operator of the grid), including the utilization of a generation source, or alternatively, the amount paid due to CO2 emissions. 

The second problem, referred to as $\mSUB$, generalizes the one studied by \cite{ElbassioniN17,KlimmPRS22} and takes the following form: 
\begin{align}
	(\mSUB)~\max_{\bx\in\{0,1\}^n} \quad& \displaystyle \bx^T A\bx \nonumber\\
	\text{\em s.t.}\quad &   \displaystyle {\sum}_{k=1}^K (\bp_k^\top\bx)^2 \le C,\label{mqclp-}
\end{align}
where $K\le n$, $\bp_k\in\mathbb{Q}_+^{n}$ and the objective is submodular, i.e., the matrix $A$ has  {\it non-positive} off-diagonal entries. Let $\ba$ be the vector of diagonal entries of $A$.  One can assume, w.l.o.g., that $\ba$ is non-negative, since if $a_i<0$ then setting $x_i=0$ does not affect the optimal solution. The constraint (\ref{mqclp-}) can be written in the form of $\bx^T Q\bx\le C$, for some positive semi-definite (PSD) matrix $Q\in\QQ_+^{n\times n}$. Since $\bx\in\{0,1\}^n$, one can write the objective function in the equivalent form of $f(\bx):= \bx^\top  A^*\bx +\ba^\top\bx$, where $A^*$ is obtained from $A$ by setting all the diagonal entries to $0$, without changing optimal solutions. 
It is not hard to verify  
that in this case the (discrete) function $f(\bx)$ (over $\{0,1\}^n$) satisfies submodularity, i.e., $f(\bx^S)+f(\bx^V)\ge f(\bx^{S\cup V}) + f(\bx^{S\cup V})$, for every two sets $S,V\subseteq \{1,2,\ldots,n\}$, where $\bx^S$ denotes the binary vector with $x_i^S=1$ if and only if $i\in S$. 

\paragraph*{Approximation algorithms} For completeness, we briefly review the definition of approximation algorithms and schemes. For $\alpha>0$, a vector $\bx\in\{0,1\}^n$ is said to be an {\em $\alpha$-approximate solution} for a maximization problem  (resp., minimization problem) with an objective $f$, if $\bx$ is a feasible solution satisfying $f(\bx)\ge \alpha\cdot\opt$ (resp., $f(\bx)\le \alpha\cdot\opt$), where $\opt$ is the value of an optimal solution. For simplicity, if $\alpha=1-\epsilon$ (resp., $1+\epsilon$), for $\epsilon>0$, we shall call $\bx$ {\it $\epsilon$-optimal}. Recall that a PTAS is an algorithm that runs in time polynomial in the input size $n$, for every fixed $\epsilon$, and outputs an $\epsilon$-optimal solution. 
A QPTAS is similar to a PTAS but the running time is quasi-polynomial (i.e., of the form $n^{\polylog n}$), for every fixed $\epsilon$. Finally, a PTAS will become an FPTAS if its running time is polynomial in $\frac{1}{\epsilon}$.

\paragraph*{Notations} We write $[n]:=\{1,2,\ldots,n\}$ for any positive integer $n$. 
For two vectors $\bx,\by$ we write $\bx\leq \by$ if $x_i\le y_i$ for all $i\in[n]$. For a vector $\bx\in\RR^n$ and a subset $S\subseteq[n]$, we write $\bx(S):=\sum_{i\in S}x_i$. Let $\|\bd\|_1=\sum_{i\in[n]}|d_i|$ and $\|\bd\|_2=\sum_{i\in[n]}d_i^2$ denote the $\ell_1$-norm and $\ell_2$-norm of vector $\bd$, respectively. We denote by $\bone$ and $\bo$ the vectors (or matrices of appropriate dimensions) of all $1$s and $0$s, respectively. 
Lastly, with a slight abuse of notation, we shall refer to the value of an optimal solution of a given instance of $\mCKP$ or $\mSUB$ problems by \opt.

\section{An Approximation Scheme for MinCKP}

In this section, we present a PTAS for $\mCKP$. As previously remarked, the relaxation with continuous variables in $[0,1]$ is non-convex and thus cannot be solved by existing techniques for convex (quadratic) programs. Interestingly, by exploiting the structure of a class of the relaxation's optimal solutions, and carefully analyzing an equivalent dual program, we can reduce the space of candidates so that an optimal solution can be found efficiently. It can then be used to derive an integer solution with a bounded guarantee. 


\begin{theorem}\label{t-mqclp}
   There is a PTAS for $\mCKP$.
\end{theorem}

Denote by  $\nlp$ the (non-convex) relaxation of $\mCKP$.  
It might be the case that the optimum solution $\bx$ for $\nlp$ is attained at the boundary of the feasible set, and thus has irrational components. Hence, by saying ``an optimal solution'' we essentially mean a {\em rational} solution $\widehat \bx\in[0,1]^n$ with $\bc^\top\widehat\bx\le \bc^\top\bx^* +\delta$ and $(\bp^\top\widehat\bx)^2 +(\bq^\top\widehat\bx)^2 \ge C-\delta$, for an arbitrary small $\delta>0$, where $\bx^*$ is an optimum solution for $\nlp$. For a given $\delta>0$, we say that $\widehat\bx$ is $\delta$-feasible for $\nlp$ if it satisfies the second inequality, and $\delta$-optimal if it satisfies both inequalities. 
 At the heart of our method is the following key result.
\begin{theorem}\label{t-nlp}
    For any given $\delta>0$, a $\delta$-optimal solution for the relaxation $\nlp$ can be computed in time $\poly(L,\log \frac{1}{\delta})$, where $L$ is the length of the binary representation of the input vectors instance.
\end{theorem}

We assume that $\|\bp\|_1^2+\|\bq\|_1^2\ge C$, so that the problem is feasible. We also assume that $\bp$ and $\bq$ are linearly independent vectors; otherwise, the problem simplifies to a linear program.
We will prove that, given a fixed positive value $T$, the problem of finding a $\delta$-feasible solution, if it exists, of value {\it exactly} $T$ can be achieved in polynomial time. We denote this problem by $\nlp[T]$.  As the objective function and the constraint~\raf{mqclp-} are both monotone\footnote{Note that, given a  feasible solution to \nlp\ with value $\bc^\top\bx<T$, we can increase the components of $\bx$, one by one, without violating the inequality~\raf{mqclp-}, until either $\bc^\top\bx$ becomes equal to $T$, or $x_i$ becomes equal to $1$ for all $i$; the latter case cannot happen unless $T\ge\|\bc\|_1$. Thus, checking if the optimum value of $\nlp$ is at most $T$ is equivalent to checking if $\nlp[T]$ is feasible.}, this result combined with a binary search (via Algorithm~\ref{binary}) yields an efficient algorithm for finding a $\delta$-optimal solution to $\nlp$, and hence proving Theorem~\ref{t-nlp}.

 \begin{algorithm}[!htb]
\caption{\textsc{Binary Search For Solving $\nlp$}} \label{binary}
\begin{algorithmic}[1]
\Require Vectors $\bc,\bp,\bq\in\QQ_+^n$, a rational number $C>0$, and an accuracy $\delta\in(0,1)$
\Ensure A $\delta$-optimal solution $\bx$  for~$\nlp$
\State $\lb \leftarrow 1$; $\ub \leftarrow \bc^\top \bone$
\While{$\ub-\lb>\delta$}
    $T = \frac{1}{2}(\ub-\lb)$
    \State Solve $\nlp[T]$ 
    \If{$\nlp[T]$ has a $\delta$-feasible solution $\bx\in[0,1]^n$}
         \State $\ub \leftarrow T$
    \Else
        \State $\lb \leftarrow T$
    \EndIf
\EndWhile
\State \Return $\bx$
\end{algorithmic}
\label{alg-min}
\end{algorithm}

\begin{lemma}
    For a given $T,\delta>0$, one can find a $\delta$-feasible solution to $\nlp[T]$ by solving $O(n^2)$ linear programs.
\end{lemma}
\begin{proof}
{\em A high-level description}. We rely on the structural property of optimal solutions of $\nlp$, which will be analyzed using the dual of the linear program~(\lp) given below. Based on the values of the (basic feasible) optimal dual solution of this linear program, we can classify the variables in an optimal solution for $\nlp$ into those having certainly a value in $\{0,1\}$ and those who are possibly taking values in $(0,1)$. Each variable in the latter set (denoted by $\Delta$ below) has the property that the corresponding objective cost can be written as a linear combination of the corresponding $p_i$ and $q_i$. This property together with the assumption that $\bc^\top\bx=T$ allows us to eliminate one of the sums $\xi:=\sum_{i\in\Delta}p_ix_i$, $\zeta:=\sum_{i\in\Delta}q_ix_i$, leading to a single quadratic inequality in only one variable $\xi$ (or $\zeta$), which can be solved exactly, thus reducing the problem of checking the feasibility of $\nlp[T]$ to solving at most $2$ linear programs. The formal procedure is given in Algorithm~\ref{dual}.  

We first give a structural property of an optimal solution of $\nlp$, which can help to check the feasibility of $\nlp[T]$ efficiently.
  Let $\bx^*$ be an optimal solution to $\nlp$. We will use $\bx^*$ as a parameter in analyzing the structure of an optimal solution.  Let $\alpha=\bp^\top\bx^*$ and $\beta=\bq^\top\bx^*$. Consider the following linear program:
    \begin{align*}
	(\lp)~\min_{\bx\in[0,1]^n} \qquad& \displaystyle \bc^\top\bx\\
	\text{\em s.t}
     \qquad &   \displaystyle \bp^\top\bx = \alpha,~~ \bq^\top\bx = \beta.
    \end{align*}
    Note that $\lp$ and $\nlp$ have the same optimal solution. Indeed, it can be seen that every feasible solution to $\lp$ is also feasible to $\nlp$. Since both programs have the same objective, every optimal solution to $\lp$ must also be optimal to $\nlp$. In what follows we make use of the dual of $\lp$ in analyzing the structure of such a solution.

    \begin{algorithm}[!htb]
    \caption{\textsc{Finding a $\delta$-Feasible  Solution}} \label{dual}
    \begin{algorithmic}[1]
    \Require Vectors $\bc,\bp,\bq\in\QQ_+^n$ and numbers $C,T\in\QQ_+$,  and an accuracy $\delta\in(0,1)$
    \Ensure Either a $\delta$-feasible solution $\bx$ to $\nlp[T]$ or declare that $\nlp[T]$ is not feasible
         \For{each pair $(j,k)\in[n]\times [n]$ with $j<k$}
            \If{$(p_j,p_k)$ and $(q_j,q_k)$ are linearly independent} 
                \State Compute $(z_1,z_2)$ such that $$p_j z_1 + q_j z_2  = c_j;~ p_k z_1 + q_k z_2  = c_k$$
                \State $\Delta_0\leftarrow \{i\in[n]~|~ p_i z_1 + q_i z_2  < c_i \}$
                \State $\Delta_1\leftarrow \{i\in[n]~|~ p_i z_1 + q_i z_2  > c_i \}$  
                \State $x_i\leftarrow 1$ for all $i\in \Delta_1$
                \State $x_i\leftarrow 0$ for all $i\in \Delta_0$ 
                \State $\Delta\leftarrow [n]\setminus \Delta_1\cup\Delta_0$ 
                \State $\xi\leftarrow\sum_{i\in \Delta} p_i x_i$
                \State $\zeta\leftarrow\sum_{i\in \Delta} p_i x_i$
                \State Reduce $\nlp[T]$ to a system of variables $\xi,\zeta$
        \If{$z_2>0$}
            \State Compute $O(\frac{\delta}{2^L})$-approximations $\tilde\xi_1\le\tilde\xi_2$  of the roots of 
            (\ref{q-eq}) in $\xi$, otherwise set $\tilde\xi_1=-\infty$;  $\tilde\xi_2=+\infty$\label{s1}
            \State $\overline \xi \gets \bp(\Delta)$ 
            \State $\Omega(\xi)\gets[0,\overline\xi] ~\bigcap \left(~(-\infty,\tilde\xi_1]\cup [\tilde\xi_2,+\infty]\right)$ 
            \For{$I\in\Omega(\xi)$}
                 \State Define LP as
                        $$  
                        \left\{\begin{matrix}
                        z_1^*\cdot\left(\sum_{i\in \Delta} p_i x_i\right) + z_2^*\cdot\left(\sum_{i\in \Delta} q_i x_i\right)  = T- \bc(\Delta_1)\\
                        \sum_{i\in \Delta} p_i x_i \in I,
                        \qquad
                        x_i\in[0,1], ~\quad \forall~i\in\Delta.
                        \end{matrix}
                        \right.
                        $$
                \If {LP has a solution $(x_i)_{i\in\Delta}$}
                        \State  \Return $\bx$
                \EndIf      
            \EndFor \label{s2}
        \Else\Comment{\emph{$z_1>0$}}
            \State Similar to lines~\ref{s1}-\ref{s2} (omitted for brevity)
        \EndIf
    \EndIf    
    \EndFor
    \State \Return ``$\nlp[T]$ is not feasible''
    \end{algorithmic}
    \end{algorithm}

    The dual of ($\lp$)  is
    \begin{align}
	\max_{\bz,\by} \qquad& \displaystyle \alpha \cdot z_1 + \beta \cdot z_2 - \bone^\top\by \nonumber\\
	\text{\em s.t.}
     \qquad &   \displaystyle  p_i\cdot z_1 + q_i\cdot z_2 - y_i\le c_i, \quad \forall ~i\in[n],\label{dual-1}\\
	   \qquad &   \displaystyle y_i\ge 0, \label{dual-2}\hspace{3.05cm} \forall ~i\in[n].
    \end{align}
    Since the primal is feasible with the primal solution $\bx^*$, then the dual is also feasible and let $(\bz^*,\by^*)$ be an optimal (basic feasible) dual solution\footnote{Note that the linear independence of $\bp$ and $\bq$ guarantees that the optimum of the dual is attained at a vertex.} corresponding to the primal solution $\bx^*$.  The optimum values of the primal-dual pair coincide by the strong duality criterion. Moreover, by the complementary slackness condition, it holds that, for all $i\in[n]$,
    \begin{description}
        \item[(i)]  $x^*_i = 0~\text{or}~  p_i\cdot z^*_1 + q_i\cdot z^*_2 - y^*_i - c_i=0$, 
       \item[(ii)] $x^*_i = 1~ \text{or}~ y^*_i = 0.$
    \end{description}
    
    
    The complementary slackness criterion implies that the primal solution $\bx^*$  satisfies the following conditions 
       \[
        x^*_i=
        \left\{
        \begin{matrix}
        0 & \text{if} & p_i\cdot z^*_1 + q_i\cdot z^*_2 < c_i,\\ 
        1 & \text{if} & p_i\cdot z^*_1 + q_i\cdot z^*_2 > c_i.
       \end{matrix}
       \right.
       \quad \text{for}~\forall~i\in[n]. \qquad   (\star)
       \]
        Note that $x^*_i$ can take any value in $[0,1]$ when the equality $p_i\cdot z^*_1 + q_i\cdot z^*_2 = c_i$ happens. 
   
        The following observation is due to the linear independence of $\bp$ and $\bq$ and the fact that $(\bz^*,\by^*)$ is assumed to be a basic feasible solution to the dual. 
        \begin{observation}
            There must be a pair $(j,k)\in[n]\times [n]$ such that the two vectors $(p_j,p_k)$ and $(q_j,q_k)$ are linearly independent and the dual vector $z^*=(z^*_1,z^*_2)$ fulfills the system of equations $\{p_j\cdot z_1 + q_j\cdot z_2  =~ c_j,~p_k\cdot z_1 + q_k\cdot z_2  =~ c_k\}$.
        \end{observation}
        Indeed, since $(\bz^*,\by^*)$ is a basic feasible solution to the dual, there must exist $n+2$ linearly independent constraints among the constraints~\raf{dual-1}-\raf{dual-2} that are tight at $(\bz^*,\by^*)$. This implies that we must have two distinct indices $j,k\in[n]$ such that the two inequalities~\raf{dual-1} and~\raf{dual-2} hold as equalities for $i=j,k$ and such that the four obtained equations are linearly independent. This leads to the claim in the observation.

    The above observation implies that $z^*_1,z^*_2$ are uniquely and fully determined by a $2\times 2$ linear system obtained by selecting two indices $j,k\in[n]$ such that the two vectors $(p_j,p_k)$ and $(q_j,q_k)$ are linearly independent. Once the solutions $z_1^*,z_2^*$ are obtained, one can classify the values of variables in the primal solution $\bx$ as follows
   
        \begin{itemize}
            \item If $p_i\cdot z^*_1 + q_i\cdot z^*_2  < c_i$, then $x_i = 0$,
            \item If $p_i\cdot z^*_1 + q_i\cdot z^*_2  > c_i$, then $x_i =1$,
            \item If $p_i\cdot z^*_1 + q_i\cdot z^*_2  = c_i$, then $x_i \in[0,1]$.
       \end{itemize}
     
        In other words, based on the values of $z^*_1,z^*_2$, we know that some of the variables in the primal solution must be equal to $1$ or $0$, but possibly not all. Define $\Delta_1,\Delta_0$ to be the sets containing indices $i$ with $x^*_i=1$ and $x^*_i=0$, respectively, and let $\Delta = [n]\setminus \Delta_1 \cup \Delta_0$. Then the values of $x^*_i$'s, $i\in\Delta$, must be the optimal to the $\nlp$ program ($\nlp[\Delta]$), parameterized by $\Delta$:
\begin{align*}
\min ~& \displaystyle {\sum}_{i\in \Delta} c_i x_i  + \bc(\Delta_1)  \\  
\text{\em s.t} ~&   \displaystyle  \left(\bp(\Delta_1) + \sum_{i\in \Delta } p_i x_i\right)^2 + \left(\bq(\Delta_1) + \sum_{i\in \Delta} q_i x_i\right)^2 \ge C,\\
& x_i\in[0,1], \quad\forall~i\in\Delta,
\end{align*}
    where $\bp(\Delta_1)={\sum}_{i\in \Delta_1} p_i,\bq(\Delta_1)={\sum}_{i\in \Delta_1} q_i,\bc(\Delta_1)={\sum}_{i\in \Delta_1} c_i$. 
 It follows that the problem of checking if $\nlp[T]$ has a feasible solution is equivalent to checking if there is a feasible solution to $\nlp[\Delta]$ with $\sum_{i\in \Delta} c_i x_i  + \bc(\Delta_1)= T$. 
By the definition of $\Delta$,  it follows that
        \[
       z_1^*\cdot\left(\sum_{i\in \Delta} p_i x_i\right) + z_2^*\cdot\left(\sum_{i\in \Delta} q_i x_i\right) =  {\sum}_{i\in \Delta} c_i x_i = T- \bc(\Delta_1).
        \]
        Based on this, the problem is now to check if the following nonlinear system is  feasible with variables  $\xi=\sum_{i\in \Delta} p_i x_i \in [0,\overline\xi], ~ \zeta =\sum_{i\in \Delta} q_i x_i\in [0,\overline\zeta]$:
         \begin{align}
        \left\{
        \begin{matrix}
        z^*_1 \cdot \xi + z_2^* \cdot \zeta &=&~ T- \bc(\Delta_1),\label{e1}\\ 
        (\bp(\Delta_1) +\xi)^2 + (\bq(\Delta_1) +\zeta)^2 &\ge &~ C.
       \end{matrix}
       \right.
       \end{align}
       where    $\overline \xi = \bp(\Delta)$, and $\overline \zeta = \bq(\Delta)$ are upper bounds on the values of $\xi$ and $\zeta$, respectively.

       Since we assume $\bc>0$, it follows that at least one of the dual solutions $z_1^*,z_2^*$ must be positive. Without loss of generality, we assume $z_2^*>0$. 
       To check the feasibility of \raf{e1}, one can compute $\zeta$ as a function of $\xi$ from the first equation, and then plug it into the second one, leading to a quadratic inequality in one variable $\xi$:
       \begin{multline}
       \label{q-eq}
       \left(z_2^*\bp(\Delta_1) +z_2^*\xi\right)^2 + \left(z_2^*\bq(\Delta_1) +T - \bc(\Delta_1) -z_1^*\xi\right)^2\\
       \ge~ (z_2^*)^2C,
       \end{multline}
       from which one can easily get the solution of $\xi$ as 
       \[
       \xi^*\in [0,\overline\xi] ~\bigcap \left(~(-\infty,\xi_1]\cup [\xi_2,+\infty]\right) \equiv \Omega(\xi),
       \]
       where $\xi_1,\xi_2$ are the two (possibly identical) roots of the quadratic equation obtained by setting the inequality in~\raf{q-eq} to an equality\footnote{Note that it is possible that the quadratic equation has no real solution in which case, we set $\xi_1=-\infty$ and $\xi_2=+\infty$; this means that our assumption about the feasibility of $\nlp[T]$ is wrong.}. Note that, depending on the signs of $\xi_1$ and $\xi_2$, $\Omega(\xi)$ consists of at most two intervals on the real line. Using each interval in $\Omega(\xi)$ and the equation in~\raf{e1}, 
       the remaining variables $x^*_i$'s, $i\in\Delta$, can be recovered by solving at most 2 linear programs, one for each choice of $I\in\Omega(\xi)$:  
        \begin{align*}
        \left\{
        \begin{matrix}
        z_1^*\cdot\left(\sum_{i\in \Delta} p_i x_i\right) + z_2^*\cdot\left(\sum_{i\in \Delta} q_i x_i\right)  = T- \bc(\Delta_1),\\
        \sum_{i\in \Delta} p_i x_i \in I,
        &&\\
        x_i\in[0,1], ~\quad \forall~i\in\Delta.
       \end{matrix}
       \right.
       \end{align*}
 Note that, due to the possible non-rationality of $\xi_1$ and $\xi_2$, we may need to work with $\delta'$-approximations of $\xi_1$ and $\xi_2$, i.e., compute numbers $\tilde\xi_1$ and $\tilde\xi_2$ such that $|\xi_1-\tilde\xi_1|\le\delta'$ and  $|\xi_2-\tilde\xi_2|\le\delta'$. Such approximations can be found in time $\poly(L,\log\frac{1}{\delta'})$, and by choosing $\delta'$ sufficiently small, we can guarantee that~\raf{q-eq} is satisfied within an additive error of $\delta$, thus yielding a $\delta$-feasible solution for $\nlp[T]$.~\end{proof}


\begin{proof} ({\bf of Theorem}~\ref{t-mqclp}). We rely on the partial enumeration method~\cite{FC84}. 
For subsets $S_0,S_1\subseteq[n]$,  let $\nlp[S_0,S_1]$ denote the restriction of $\nlp$ when we enforce $x_i=0$ for $i\in S_0$ and $x_i=1$ for $i\in S_1$. It is straightforward to modify Algorithm~\ref{alg-min} to find a $\delta$-optimal solution for $\nlp[S_0,S_1]$.
Given the desired accuracy $\epsilon\in(0,1)$, we consider all subsets $S_1$ of size at most $h:=\lceil\frac2\epsilon\rceil$ and let $S_0:=\{i\in[n]:~c_i>\min_{j\in S_1}c_j\}$. For each considered subset $S_1$, we solve $\nlp[S_0,S_1]$ and round the $\delta$-optimal solution $x$ obtained as described below. Among all the rounded solutions obtained, we return the solution with minimum cost. Effectively, we are guessing the $h$ indices with the highest $c_i$ values in an optimal solution of $\mCKP$ and solving for the remaining items. Let us denote the corresponding sets in the optimal guess by $S_0^*$ and $S_1^*$. 
 Let $\bx^*$ be a $\delta$-optimal solution to $\nlp[S_0^*,S_1^*]$,   $\alpha=\bp^\top\bx^*$,  $\beta=\bq^\top\bx^*$, and consider the following linear program ($\lp'[S_0,S_1]$):
    \begin{align*}
	\min_{\bx\in[0,1]^n} \qquad& \displaystyle \bc^\top\bx\\
	\textit{s.t}
     \qquad &   \displaystyle \bp^\top\bx \ge \alpha,\quad \bq^\top\bx \ge \beta,\\
     \qquad &\text{$x_i=0$, ~~for~~ $i\in S_0$},\\
     \qquad &\text{$x_i=1$, ~~for~~ $i\in S_1$.}
    \end{align*}
A basic fact from Linear Programming tells us that every bounded linear program with only two non-trivial constraints has an optimal (basic feasible) solution lying at some vertex of the polytope defining the constraints, which cannot have more than two fractional components.
\begin{observation}\label{lem:BFS}
    There is an optimal solution to $\lp'[S_0,S_1]$ with at most two fractional entries.
\end{observation}
Let $z^*$ and $z'$ be the optimal values for $\nlp[S_0^*,S_1^*]$ and $\lp'[S_0^*,S_1^*]$, respectively. Then, as $x^*$ is both $\delta$-optimal for $\nlp[S_0^*,S_1^*]$ and feasible for $\lp'[S_0^*,S_1^*]$, we must have $z'\le z^*+\delta$.  Hence, there exists a basic optimal solution of $\lp'[S_0^*,S_1^*]$, with at most $2$ fractional components (which can be found efficiently), which is also a $\delta$-optimal solution of $\nlp[S_0^*,S_1^*]$. By rounding each of these two fractional components to $1$, we obtain a $\delta$-feasible solution $\widehat\bx$ to $\mCKP$ of value $\bc^\top\widehat \bx$ at most 
\[
z^*+\delta+2\min_{j\in S_1^*}c_j\le z^*+\delta+\frac{2\sum_{j\in S_1^*}c_j}{h}\le (1+O(\epsilon))\opt,
\]
if we choose $\delta=\epsilon\cdot2^{-2L-1}\le \epsilon \cdot\min_{j\in [n]}c_j$. Observe that, for this choice of $\delta$, $(\bp^\top\widehat \bx)^2+(\bq^\top\widehat \bx)^2\ge C-\delta$ implies that $(\bp^\top\widehat \bx)^2+(\bq^\top\widehat \bx)^2\ge C$, since $C-((\bp^\top\widehat \bx)^2+(\bq^\top\widehat \bx)^2)$ is a rational number of value at least $2^{-2L}>\delta$. We conclude that $\widehat\bx$ is feasible solution to $\mCKP$ of value at most $(1+O(\epsilon))\opt$. The running time is $n^{\frac2\epsilon+O(1)}L^{O(1)}$.
\end{proof}

\section{Pipage Rounding for MaxSUB}
 %
In this section we provide a pipage rounding procedure for MaxSUB, given a (fractional) solution of its relaxation (denoted as $\cF$, obtained by changing binary variables to continuous ones). 
Our aim is to prove the following result.

\begin{theorem}\label{thm:sub-convex-quadratic}
There is an $\frac{\alpha\phi^2}{2}$-approximation algorithm for $\mSUB$, provided that there is an $\alpha$-approximation algorithm for its continuous relaxation, where $\phi=\frac{2}{1+\sqrt{5}}\approx 0.618$.
\end{theorem}

To prove the Theorem~\ref{thm:sub-convex-quadratic}, we construct an algorithm, which is formally presented in Algorithm~\ref{alg:general-convex}, and prove its correctness. Before doing so, (inspired by~\cite{KlimmPRS22}) let us introduce two relaxations of $\mSUB$, which are needed in the performance analysis of the algorithm later on. In both relaxations, we replace the binary variables $\bx\in\{0,1\}^n$ by $\bx\in[0,1]^n$. Recall that $f(\bx)=\bx^\top A^*\bx + \ba^\top\bx$, define
\begin{align*}
    (\cF_1):~ & {\max}_{\bx\in[0,1]^n} \{f(\bx)~|~\bx^\top Q\bx\le C,~ \bq^\top\bx \le C  \},\\
    (\cF_2):~ & {\max}_{\bx\in[0,1]^n} \{ f(\bx)~|~ \bx^\top Q^*\bx +\bq^\top\bx \le C\},
\end{align*}
where $\bq$ is the vector containing all the diagonal entries of $Q$. It is worth noting that, because of the non-negativity of $Q$, $\cF$ and $\cF_1$ are equivalent.
\begin{lemma}\label{lem:phi-relxation}
    Every feasible integer solution to the original problem is feasible to both relaxations. Also, for every feasible solution $\bx$ to $\cF_1$, $\phi\cdot\bx$ forms a feasible solution to $\cF_2$. Furthermore, the objective value of $\phi\cdot\bx$ is at least $\phi^2$ times the optimum of $\cF_1$.
\end{lemma}
\begin{proof}
    The first claim is easy to see. For the second one, note that
\begin{align*}
(\phi\bx)^\top Q^* (\phi\bx) + \bq^\top(\phi\bx) =~& \phi^2\bx^\top Q \bx + \phi\bq^\top\bx\\
\le~& \phi^2\cdot C + \phi\cdot C =C,
\end{align*}
where the last equality follows from the definition of $\phi$. It remains to show the last claim. We have that
$
(\phi\bx)^\top A^* (\phi\bx) + \ba^\top(\phi\bx) =\phi^2(\bx^\top A^* \bx) + \phi\ba^\top\bx\ge \phi^2\cdot (\bx^\top A^* \bx + \ba^\top\bx),
$
since $\ba^\top\bx\ge0$, and $\phi<1$.
\end{proof}

\subsubsection{Performance analysis of Algorithm~\ref{alg:general-convex}.} 

Let $g(\bx):=\bx^\top Q^*\bx + \bq^\top\bx$. Suppose that $\bx$ is the solution returned at Step~\ref{step:greedy-algorithm} of Algorithm~\ref{alg:general-convex}. Then, it holds that $f(\bx)\ge\alpha f(\tilde\bx) \ge \alpha f(\bx^*)$, where $\bx^*$ and $\tilde\bx$ are, respectively, optimal solutions of $\mSUB$ and its relaxation $\cF_1$. By Lemma~\ref{lem:phi-relxation}, we have that $\by=\phi\cdot\bx$ is feasible to $(\cF_2)$. By the correctness of the pipage rounding, to be proved below, $O(n)$ applications of Algorithm~\ref{pipage-rounding} result in a solution $\by'$ such that $f(\by')\ge f(\by)$ and $g(\by')\le g(\by) \le C.$ Hence, the rounded solution $\overline\by$ is feasible for the original problem $\mSUB$.

Now let $\overline\bx$ be the integer solution returned by Algorithm~\ref{alg:general-convex}. Clearly, $f(\overline\bx)\ge f(\overline\by)$. If $f(\overline\by)\ge f(\by')$ then we are done. Otherwise, suppose that $f(\overline\by)\le f(\by')\le f(\overline\by')$, where $\overline\by'$ is obtained by rounding up the fractional component in $\by'$, says $y'_i$, to $1$. By the submodularity of $f$, it holds that
$$
 f(\overline\by')\le f(\overline\by) + f(\bx^{\{i\}})\le f(\overline\by) + f(\bx^{\{i^*\}}),
$$  where $i^*$ is defined as in Step~\ref{item-i*} of Algorithm~\ref{alg:general-convex}. \footnote{Recall that $\bx^S$ denotes the binary vector with $x_i^S=1$ if and only if $i\in S$. }

By the definition of $\overline\bx$, it holds that 
$$
f(\overline\bx)\ge \frac{1}{2}(f(\overline\by) + f(\bx^{\{i^*\}})) \ge \frac{1}{2}  f(\overline\by')\ge \frac{1}{2} f(\by').
$$
Furthermore, by the definition of $\by'$ and $\by$ and Lemma~\ref{lem:phi-relxation}, we must have that 
$$f(\by')\ge f(\by)\ge \phi^2\cdot f(\bx) \ge \phi^2\alpha\cdot f(\bx^*).$$
Hence, $f(\overline\bx)\ge \frac{\alpha\phi^2}{2}\cdot f(\bx^*)$. 



\begin{algorithm}[!htb]
\caption{\textsc{Approx. Algorithm for  $\mSUB$}} \label{alg:general-convex}
\begin{algorithmic}[1]
\Require An $\alpha$-approximation algorithm $\cA$ for the relaxation of $\mSUB$
\Ensure An integer solution $\overline \bx$ to $\mSUB$
\State Find a (fractional) solution $\bx$ to the relaxation $\cF_1$ (or $\cF$), using algorithm $\cA$ \label{step:greedy-algorithm}
\State $\by\leftarrow\phi\cdot\bx$
\While{there is a pair $(y_i,y_j)\in (0,1)^2$}
      \State Apply Algorithm~\ref{pipage-rounding} to round $(y_i,y_j)$
\EndWhile
\State Let $\by'$ be the solution with at most one fractional component 
\State Set the (single) fractional component in $\by'$ (if any) to $0$ to get an integer solution $\overline\by$
\State $i^*\leftarrow \arg\max_{i\in[n]}\{a_{i}\}$ \label{item-i*}
 \State $\overline\bx \leftarrow \arg\max \{f(\overline\by), f(\bx^{\{i^*\}})\}$
\State \Return $\overline\bx$
\end{algorithmic}
\end{algorithm}  

\begin{algorithm}[!htb]
\caption{\textsc{Pipage-Rounding-Procedure}} \label{pipage-rounding}
\begin{algorithmic}[1]
\Require A pair $(x_i,x_j)\in (0,1)^2$; $\varphi= a_{ij}x_ix_j+a'_ix_i+a'_jx_j;~
    \psi = q_{ij}x_ix_j+q'_ix_i+q'_jx_j$.
\Ensure A pair $(\overline x_i,\overline x_j)\in [0,1]^2$ with at most one fractional component  
\State $\varphi'_i\leftarrow a_{ij}x_j+a_i'$; $\varphi'_j\leftarrow a_{ij}x_i+a_j'$
\State $\psi'_i\leftarrow q_{ij}x_j+q_i'$; $\psi'_j\leftarrow q_{ij}x_i+q_j'$
\If {$\varphi\le 0$} {$\overline x_i\leftarrow 0$ and $\overline x_j\leftarrow 0$}  \EndIf
\If {$a'_i\le 0$ or $\varphi'_{i}\le 0$} {$\overline x_i\leftarrow 0$}  \EndIf
\If {$a'_j\le 0$ or $\varphi'_{j}\le 0$} {$\overline x_j\leftarrow 0$}  \EndIf
\If {$q_{ij}+q_i'=0$} {$\overline x_i\leftarrow \argmax_{x_i\in\{0,1\}}\varphi(x_i,x_j)$}\EndIf
\If {$q_{ij}+q_j'=0$} {$\overline x_j\leftarrow \argmax_{x_j\in\{0,1\}}\varphi(x_i,x_j)$}\EndIf
\If{$\varphi,\varphi'_{i},\varphi'_{j},a_i',a_j'>0$, $q_{ij}+q_i'>0$ and $q_{ij}+q_j'>0$}
\If{$\varphi \le a_j'$ and $\dfrac{\varphi'_i}{\psi'_i} \le \dfrac{a_j'}{q_j'}$} \Comment{\emph{Case 1}}
\State $\overline x_i\leftarrow 0$, $\overline x_j\leftarrow \min\left\{x_j+\dfrac{x_i\psi'_i}{q_j'},1\right\}$
    \EndIf
\If{$\varphi \le a_i'$ and $\dfrac{\varphi'_j}{\psi'_j} \le \dfrac{a_i'}{q_i'}$}\Comment{\emph{Case 2}}
\State $\overline x_i\leftarrow \min\left\{x_i+\dfrac{x_j\psi'_j}{q_i'},1\right\}$, $\overline x_j\leftarrow 0$
    \EndIf
\If{$\psi \ge q_i'$ and  $\dfrac{\varphi'_i}{\psi'_i} \ge \dfrac{a_j'+a_{ij}}{q_j'+q_{ij}}$} \Comment{\emph{Case 3}}
\If{$a_{ij}+a'_j\le 0$}{
        $\overline x_i\leftarrow 1$, $\overline x_j\leftarrow 0$}
    \Else{$\overline x_i\leftarrow 1,\overline x_j\leftarrow \max\left\{x_j-\dfrac{(1-x_i)\varphi'_i}{a_{ij}+a'_j},0\right\}$}
    \EndIf
  \EndIf
\If{$\psi \ge q_j'$ and $\dfrac{\varphi'_j}{\psi'_j} \ge \dfrac{a_i'+a_{ij}}{q_i'+q_{ij}}$} \Comment{\emph{Case 4}}
 \If{$a_{ij}+a'_i\le 0$}
        {$\overline x_i\leftarrow 0$, $\overline x_j\leftarrow 1$}
\Else{    
 $\overline x_j\leftarrow 1,\overline x_i\leftarrow \max\left\{x_i-\dfrac{(1-x_j)\varphi'_{j}}{a_{ij}+a'_i},0\right\}$}
    \EndIf
\EndIf
\EndIf
\State \Return $(\overline x_i,\overline x_j)$
\end{algorithmic}
\end{algorithm}

\subsection*{Correctness of the pipage rounding (Algorithm~\ref{pipage-rounding})}
We prove that at the end of the pipage rounding procedure at most one component of the resulting solution $\by'$ is fractional. To this end, we show that, at each step of the rounding, at least one of the fractional variables is rounded to either $0$ or $1$, without decreasing $f(\bx)$ or increasing $g(\bx)$. Note that the value of a fractional variable and its coefficient (which can be a function of other variables) may be changed after each rounding step. In addition, once a variable has been rounded to $0$ or $1$, it will never be considered in the next steps. In particular, the while loop in Algorithm \ref{alg:general-convex} runs for at most $n$ iterations. Suppose that we are at step $k$ of the rounding procedure, with two fractional variables $x^{(k)}_i,x^{(k)}_j\in(0,1)$. For ease of presentation, we denote by $\varphi(x^{(k)}_i,x^{(k)}_j)$ and $\psi(x^{(k)}_i,x^{(k)}_j)$ the parts of $f$ and $g$, respectively, that contain these two variables. Formally, let
\begin{align*}
    \varphi(x^{(k)}_i,x^{(k)}_j) =~& a_{ij}x^{(k)}_ix^{(k)}_j+a^{(k)}_ix^{(k)}_i+a^{(k)}_jx^{(k)}_j,\\
    \psi(x^{(k)}_i,x^{(k)}_j) =~& q_{ij}x^{(k)}_ix^{(k)}_j+q^{(k)}_ix^{(k)}_i+q^{(k)}_jx^{(k)}_j,
\end{align*}
be the functions of $x^{(k)}_i$ and $x^{(k)}_j$, where $a_{ij}\le 0$, $q_{i,j},q^{(k)}_i,q^{(k)}_j\ge 0$. We may use   $\varphi'_i,\varphi'_j,\psi'_i,\psi'_j$ to denote the first order derivatives of $\varphi$ and $\psi$. To simplify the notation, we drop the superscript ``$k$'' and use $a_i'$, $a_j'$, $q_i'$ and $q_j'$ to denote $a_i^{(k)}$,$a_j^{(k)}$, $q_i^{(k)}$ and $q_j^{(k)}$ respectively. It suffices to prove that one can round the two variables $x_i,x_j$ such that at least one of them goes down to $0$ or up to $1$, without decreasing  $\varphi(x_i,x_j)$ or increasing $\psi(x_i,x_j)$.

\begin{itemize}
    \item If $\varphi(x_i,x_j)\le 0$ then rounding down both variables $x_i,x_j$ to $0$ does not decrease $\varphi(x_i,x_j)$ or increase $\psi(x_i,x_j)$. Thus, we may assume that $\varphi(x_i,x_j)>0$.
    \item If $a_i'\le 0$ then set $\overline x_i=0$, and similarly for $a_j$. Thus, we may assume that $a_i',a_j'>0$.
    \item If $\varphi'_i:=a_{ij}x_j+a_i'\le 0$ then set $\overline x_i= 0$, and similarly for $\varphi'_j:=a_{ij}x_i+a_j'$. Thus, we may assume that $\varphi'_i>0$ and $\varphi'_j>0$.
    \item If $q_{ij}=q_i'=0$, we set $\overline x_i\in\{0,1\}$ so as to maximize $\varphi(x_i,x_j)$ (which is linear in $x_i$). This guarantees that $\varphi(\overline x_i,x_j)\ge\varphi(x_i,x_j)$ while $\psi(\overline x_i,x_j)=\psi(x_i,x_j)$. Thus we may assume that $q_{ij}+q_i'>0$ and similarly $q_{ij}+q_j'>0$.
\end{itemize}

Note that $\psi(x_i,x_j),\psi'_i(x_i,x_j),\psi'_j(x_i,x_j)\ge 0$. The following lemma completes the proof of the correctness of the pipage rounding. 
\begin{lemma}\label{lem:pipage-round-correctness}
    Suppose that $a_i',a_j'>0$, and $\varphi(x_i,x_j)$ together with its first derivatives  $\varphi'_i(x_i,x_j)$ and $\varphi'_j(x_i,x_j)$ and all strictly positive. Then there always exist two parameters $\alpha\in[-x_i,1-x_i],\beta\in[-x_j,1-x_j]$ so as to satisfy simultaneously both inequalities $\varphi(x_i+\alpha,x_j+\beta)\ge \varphi(x_i,x_j) $ and $\psi(x_i+\alpha,x_j+\beta)\ge \psi(x_i,x_j)$, with at least one of $\alpha,\beta$ being held at the boundary of its domain.
\end{lemma}

\begin{proof}
We have that 
\begin{align}\label{eq:system-feasibility}
    \left\{
    \begin{matrix}
    \varphi(x_i+\alpha,x_j+\beta)\ge \varphi(x_i,x_j), \\ 
    \psi(x_i+\alpha,x_j+\beta)\le \psi(x_i,x_j),
    \end{matrix}
    \right.\\
    \Longleftrightarrow
    \left\{
    \begin{matrix}
    a_{ij}\alpha\beta +\varphi'_i\alpha + \varphi'_j\beta \ge 0, \\ 
    q_{ij}\alpha\beta +\psi'_i\alpha + \psi'_j\beta \le 0,
    \end{matrix}
    \right.
\end{align}
where $\psi'_i:=q_{ij}x_j+q_i'$ and $\psi'_j:=q_{ij}x_i+q_j'$. 
We will fix one of the variables $\alpha$ or $\beta$ to one of its two boundary values, and assign the other variable a value in its range such that (\ref{eq:system-feasibility}) is satisfied. To this end, we consider four possible cases below.

\begin{itemize}
    \item {\bf Case 1:} $\alpha=-x_i$. Plugging this into (\ref{eq:system-feasibility}) implies 
    $
    \frac{\varphi'_i}{a_j'} \le \beta \le \frac{\psi'_i}{q_j'},
    $ 
which are valid inequalities (even if $q_j=0$) for some $\beta\in[-x_j,1-x_j]$ if and only if the following system is feasible
  \[
   \text{(I1)} ~ \left\{  \varphi(x_i,x_j) \le a_j',~ \text{(I1.1)}
   ~ \bigwedge~
   \dfrac{\varphi'_i}{\psi'_i} \le \dfrac{a_j'}{q_j'}.~ \text{(I1.2)}\right\}
  \] 
   \item {\bf Case 2:} $\beta=-x_j$. Using a similar argument as in Case 1, there is a value of $\alpha\in[-x_i,1-x_i]$ satisfying~\raf{eq:system-feasibility} if and only if the following system is feasible
  \[
   \text{(I2)} ~ \left\{  \varphi(x_i,x_j) \le a_i', ~\text{(I2.1)}
   ~ \bigwedge~
   \dfrac{\varphi'_j}{\psi'_j} \le \dfrac{a_i'}{q_i'}.~\text{(I2.2)}\right\}
  \]
   \item {\bf Case 3:} $\alpha=1-x_i$. Plugging this into (\ref{eq:system-feasibility}) implies 
   \begin{align}\label{eq:case-3-system}
    \left\{
    \begin{matrix}
    (a_{ij}+a_j')\beta +(1-x_i)\varphi'_i \ge 0, \\ 
    (q_{ij}+q_j')\beta +(1-x_i)\psi'_i \le 0.
    \end{matrix}
    \right.
\end{align}
Without loss of generality $a_{ij}+a_j'\ne 0$  (otherwise, it becomes only easier to find a feasible value of $\beta$).  Using similar calculations as in the previous cases, it can be shown (regardless of the sign of $a_{ij}+a'_j$) that  there is $\beta\in[-x_j,1-x_j]$ for which (\ref{eq:case-3-system}) is satisfied if and only if the following system is feasible
\[
   \text{(I3)} ~ \left\{  \psi(x_i,x_j) \ge q_i', ~ \text{(I3.1)}
   ~ \bigwedge~
   \dfrac{\varphi'_i}{\psi'_i} \ge \dfrac{a_j'+a_{ij}}{q_j'+q_{ij}}.~ \text{(I3.2)}\right\}
  \]
\item {\bf Case 4:} $\beta=1-x_j$. Similar to Case 3, there is a value of $\alpha\in[-x_i,1-x_i]$ satisfying~\raf{eq:system-feasibility} if and only if the following system is feasible
\[
   \text{(I4)} ~ \left\{ \psi(x_i,x_j) \ge q_j', ~ \text{(I4.1)}
   ~ \bigwedge~
   \dfrac{\varphi'_j}{\psi'_j} \ge \dfrac{a_i'+a_{ij}}{q_i'+q_{ij}}.~ \text{(I4.2)}.\right\}
  \]
\end{itemize}

We now prove, via contradiction, that at least one of the systems (I1), (I2), (I3), or (I4) is feasible. Before doing so, we observe some logical relations between some of the inequalities in these systems.
In the following, for an inequality (\text{I})$\in\{$(I1.1),\dots,(I4.2)$\}$, we overload notation by writing (\text{I}) to mean that the inequality holds, and use ($\overline{\text{I}}$) to denote that the reverse inequality holds. 
The following observations are easy to verify:
\begin{itemize}
\item[L1:] $\big[$(I1.2) $\vee$ (I2.2)$\big]$: indeed, suppose that this is not the case,  that is
\begin{align}
    \left\{
    \begin{matrix}
    q_j'\varphi'_i>a_j' \psi'_i \\ 
    q_i'\varphi'_j>a_i' \psi'_j.
    \end{matrix}
    \right.\\
    \Longleftrightarrow
    \left\{
    \begin{matrix}
    q_j'a_{ij}x_j + a_iq_j > a_j'q_{ij}x_j + a_j'q_i' \\ 
    q_i'a_{ij}x_i + a_jq_i > a_i'q_{ij}x_i + a_i'q_j'
    \end{matrix}
    \right.
\end{align} 
It follows that $a_{ij}(q_i'x_i+q_j'x_j)> q_{ij}(a_i'x_i+a_j'x_j)$, which does not hold as $a_{ij}
\le 0$. 

\item[L2:] Similarly, we have [(I3.2) $\vee$ (I4.2)]: indeed, suppose that this is not the case. Then,
\begin{align*}
\dfrac{\varphi'_i}{\varphi'_i} <&~ \dfrac{a_j'+a_{ij}}{q_j'+q_{ij}}\\ \Longleftrightarrow  a_{ij}q_j'x_j+a_iq_j+a_i'q_{ij}<& ~a_j'q_{ij}x_j+q_i'a_j'+a_{ij}q_i'
\end{align*}
and
\begin{align*}
\dfrac{\varphi'_j}{\psi'_j} <&~ \dfrac{a_i'+a_{ij}}{q_i'+q_{ij}}\\ \Longleftrightarrow a_{ij}q_i'x_i+a_jq_i+a_j'q_{ij}<&~ a_i'q_{ij}x_i +a_i'q_j+a_{ij}q_j'.
\end{align*}
Collecting terms, we get
\[
a_{ij}\left[
q_i'(1-x_i)+q_j'(1-x_j)
\right]>q_{ij}\left[a_i'(1-x_i)+a_j'(1-x_j)
\right],
\]
which is a contradiction as $a_{ij} \le 0$ and $x_i,x_j<1$. 

\item[L3:]  $\big[$($\overline{\text{I1.1}}$) $\wedge$ (I1.2)$\big ]$ $\Longrightarrow$ (I4.1): indeed, the antecedent of this implication imposes that
\begin{align*}
    a_j'(1-x_j)<a_{ij}x_ix_j+a_i'x_i=x_i\varphi'_i\le x_i\cdot\frac{a_j'\psi'_i}{q_j'},
\end{align*}
giving $q_j'<q_j'x_j+x_i\psi'=\psi(x_i,x_j)$, which satisfies (I4.1).
\item [L4:] Similarly, $\big[$(${\text{I2.2}}$) $\wedge$ ($\overline{\text{I3.1}}$)$\big]$ $\Longrightarrow$ (I2.1): indeed, the antecedent of the implication imposes that
\begin{align*}
    q_i'(1-x_i)>q_{ij}x_ix_j+q_j'x_j=x_j\psi'_j\ge x_j\cdot\frac{q_i'\varphi'_j}{a_i'},
\end{align*}
implying that $q_i'>0$, and consequently giving $a_i'>a_i'x_i+x_j\varphi'_j=\varphi(x_i,x_j)$, which satisfies (I2.1).
\item[L5:] $\big[$(I2.2) $\vee$ (I4.2)$\big]$: indeed, suppose that this is not the case. Then we get a contradiction:
\[
  \dfrac{a_i'}{q_i'}\ge\dfrac{a_i'+a_{ij}}{q_i'+q_{ij}}>\dfrac{\varphi'_j}{\psi'_j} > \dfrac{a_i'}{q_i'},
\]
as $a_{ij}\le 0$, and $a_{ij},a_i',q_i',q_{ij}\ge 0$.
\end{itemize}

Now, suppose that none of the systems (I1), (I2), (I3), and (I4) is feasible. W.l.o.g. assume that (I1.2) holds (the case when (I2.2) holds can be done similarly by exchanging the roles of the indices $i$ and $j$). Then we get the following sequence of implications:
\begin{align*}
({\text{I1.2}})~\Longrightarrow~(\overline{\text{I1.1}})~\stackrel{L3}{\Longrightarrow} ~(\text{I4.1})
~\Longrightarrow~(\overline{\text{I4.2}})\\
~\stackrel{L2}{\Longrightarrow} ~(\text{I3.2})~{\Longrightarrow}~(\overline{\text{I3.1}})~\stackrel{L4}{\Longrightarrow} ~(\overline{\text{I2.2}})~\stackrel{L5}{\Longrightarrow} ~(\text{I4.2}),
\end{align*}
leading to the contradiction that the system (I4) is feasible.~\end{proof}

Theorem~\ref{thm:sub-convex-quadratic} paves the way for designing approximation algorithms for $\mSUB$, based on solving (approximately or exactly) its relaxation $\cF$. We are not aware of any algorithm for efficiently solving such a relaxation. However, for general convex constraint (where $Q$ is a general non-negative PSD matrix), one can show a constant approximation algorithm by taking advantage of the result of \cite{BianL0B17}.

\begin{theorem}[\cite{BianL0B17}]\label{t-dr}
    There is an iterative greedy algorithm for maximizing a (continuous) DR-submodular function $f$ subject to a down-closed convex set $S$, such that it proceeds in $K$ steps (each taking polynomial time), and outputs a solution $\bx$ such that $f(\bx)\ge \frac{1}{e}f(\bx^*) -\frac{LD^2}{2K}-O(\frac{1}{K^2})f(\bx^*)$, where $\bx^*$ is an optimal solution, $L$ is a constant such that $\|\nabla f(\bx)-\nabla f(\by)\|_2\le L\|\bx-\by\|_2$, and $D=\max_{\bx,\by\in S} \|\bx-\by\|_2$. 
\end{theorem}

Here, a convex set $S$ is said to be {\em down-closed} if for any $\bx,\by\in\RR_+^n$,  $\bx\le \by$ and  $\by\in S$ implies that $\bx\in S$. Also, a twice-differentiable function $f:\cX\rightarrow \mathbb{R}$ is DR-submodular if and only if $\nabla^2_{ij} f(\bx) \le 0$ for $\forall\bx\in\cX$ and $i,j\in[n]$ (see \cite{BianMB017} for more details). When applied to our problem, the above result gives the following.
\begin{corollary}\label{cor:1e-approx}
    There is a $(\frac{1}{e}-O(\epsilon))$-approximation algorithm for maximizing  $f(\bx)$ under the constraints $ \bx^\top Q\bx\le C,~ \bq^\top\bx \le C, ~\bx\in[0,1]^n$, when $A^*\le\bo$, $Q\ge \bo$ and $Q$ is PSD.
\end{corollary}

\begin{proof}
         Clearly, the convex set defined by $\{\bx\in[0,1]^n:~\bx^\top Q\bx\le C,\bq^\top\bx \le C\}$ with non-negative PSD matrix $Q$, and positive number $C$, is down-closed. We now derive upper bounds on the values of the parameters $L$ and $D$ defined in Theorem~\ref{t-dr}. The gradient of $f$ can be computed as $\nabla f(\bx)= A^*\bx + \ba$, and thus
     \begin{align*}
     \|\nabla f(\bx)-\nabla f(\by)\|_2 = \|A^*(\bx-\by)\|_2 \le &~\|A^*\|_F\cdot\|\bx-\by\|_2\\
     &~\le n\lambda\cdot\|\bx-\by\|_2,
     \end{align*}
     where $\|A^*\|_F$ denotes the Frobenius norm of $A^*$, $\lambda$ is the maximum absolute value of entries in $A^*$, and the first inequality follows from the Cauchy-Schwarz inequality. This gives $L\le n\lambda$. On the other hand, in our case of $\mSUB$, $D$ is upper bounded by the maximum length of any feasible point in the hypercube $[0,1]^n$, which equals $\sqrt{n}$. To show the approximation factor, it is noted that $f(\bx^*)\ge a_i$, for all $i\in[n]$, because of the feasibility of any unit vector $\bx^{\{i\}}$ (by the non-negativity assumption on $Q$). Also, by the non-negativity of the objective function, it holds that $f(\bx^{\{i,j\}})=a_i+a_j+a_{ij}+a_{ji}\ge 0$ for every $(i,j)\in[n]^2$, which implies that $|a_{ij}|, |a_{ji}|\le 2\max\{a_i,a_j\}$. Therefore, $f(\bx^*)\ge \frac{1}{2}\lambda$. By plugging in $K:=\frac{n^{2}}{\epsilon}$ in Theorem~\ref{t-dr}, one can get a $(\frac{1}{e}-O(\epsilon))$ approximation factor as desired.~\end{proof}
     
Now by applying Corollary~\ref{cor:1e-approx} and Theorem~\ref{thm:sub-convex-quadratic}, one can obtain the following result for $\mSUB$. Note that the objective function $f(\bx)= \bx^\top  A\bx +\ba^\top\bx$ can be written as $\bx^\top  A^*\bx +\ba^\top\bx$ for binary variables, which exhibits DR-submodularity over $[0,1]^n$ as $A^*<\bo$.

\begin{theorem}\label{thm:sub-convex-quadratic-1}
There is an $(\frac{\phi^2}{2e}-O(\epsilon))$-approximation algorithm for $\mSUB$. 
\end{theorem}

\section{Conclusion}
We have presented novel approximation algorithms for two power scheduling problems with the sum-of-squares constraints, one of them significantly improved upon the previous known results. Our work also leads to several open problems for future work. Perhaps the most interesting question is whether our PTAS for $\mCKP$  can be extended to the case when the constraint involves any constant number of squares. One can observe that our approach to solving the natural (non-convex) relaxation cannot be extended in this setting, and thus, the development of new techniques may be required. For the submodular maximization problem $\mSUB$, it would be interesting to investigate the following two questions: i) Is there a better approximation algorithm for solving the relaxation? and ii) Is it possible to have pipage rounding that works for more general (differential) DR-submodular functions? For the first question, one idea is to utilize the special structure of the sum-of-square constraints, for which the semi-definite program-based method might be applied. Finally, extending all the obtained results in this paper to a setting with more than one constraint would also be another promising direction for research.

\bibliography{main}

\begin{thebibliography}{17}
\providecommand{\natexlab}[1]{#1}

\bibitem[{Bansal, Kimbrel, and Pruhs(2004)}]{BansalKP04}
Bansal, N.; Kimbrel, T.; and Pruhs, K. 2004.
\newblock Dynamic Speed Scaling to Manage Energy and Temperature.
\newblock In \emph{45th Symposium on Foundations of Computer Science {(FOCS} 2004), 17-19 October 2004, Rome, Italy, Proceedings}, 520--529. {IEEE} Computer Society.

\bibitem[{Bian et~al.(2017{\natexlab{a}})Bian, Levy, Krause, and Buhmann}]{BianL0B17}
Bian, A.; Levy, K.~Y.; Krause, A.; and Buhmann, J.~M. 2017{\natexlab{a}}.
\newblock Non-monotone Continuous {DR}-submodular Maximization: Structure and Algorithms.
\newblock In Guyon, I.; von Luxburg, U.; Bengio, S.; Wallach, H.~M.; Fergus, R.; Vishwanathan, S. V.~N.; and Garnett, R., eds., \emph{Advances in Neural Information Processing Systems 30: Annual Conference on Neural Information Processing Systems 2017, December 4-9, Long Beach, CA, {USA}}, 486--496.

\bibitem[{Bian et~al.(2017{\natexlab{b}})Bian, Mirzasoleiman, Buhmann, and Krause}]{BianMB017}
Bian, A.~A.; Mirzasoleiman, B.; Buhmann, J.~M.; and Krause, A. 2017{\natexlab{b}}.
\newblock Guaranteed Non-convex Optimization: Submodular Maximization over Continuous Domains.
\newblock In Singh, A.; and Zhu, X.~J., eds., \emph{Proceedings of the 20th International Conference on Artificial Intelligence and Statistics, {AISTATS} 2017, Fort Lauderdale, FL, {USA}}, volume~54 of \emph{Proceedings of Machine Learning Research}, 111--120. {PMLR}.

\bibitem[{Chau, Elbassioni, and Khonji(2016)}]{ChauEK16}
Chau, C.; Elbassioni, K.~M.; and Khonji, M. 2016.
\newblock Truthful Mechanisms for Combinatorial Allocation of Electric Power in Alternating Current Electric Systems for Smart Grid.
\newblock \emph{{ACM} Trans. Economics and Comput.}, 5(1): 7:1--7:29.

\bibitem[{Elbassioni, Karapetyan, and Nguyen(2019)}]{ElbassioniKN19}
Elbassioni, K.~M.; Karapetyan, A.; and Nguyen, T.~T. 2019.
\newblock Approximation schemes for r-weighted Minimization Knapsack problems.
\newblock \emph{Ann. Oper. Res.}, 279(1-2): 367--386.

\bibitem[{Elbassioni and Nguyen(2017)}]{ElbassioniN17}
Elbassioni, K.~M.; and Nguyen, T.~T. 2017.
\newblock Approximation algorithms for binary packing problems with quadratic constraints of low cp-rank decompositions.
\newblock \emph{Discret. Appl. Math.}, 230: 56--70.

\bibitem[{Frieze and Clarke(1984)}]{FC84}
Frieze, A.; and Clarke, M. 1984.
\newblock Approximation Algorithm for the m-Dimensional 0-1 Knapsack Problem.
\newblock \emph{European Journal of Operational Research}, 15: 100--109.

\bibitem[{Grainger and Stevenson(1994)}]{Grainger}
Grainger, J.; and Stevenson, W. 1994.
\newblock \emph{Power System Analysis}.
\newblock McGraw-Hill.

\bibitem[{Karapetyan et~al.(2018)Karapetyan, Khonji, Chau, Elbassioni, and Zeineldin}]{KarapetyanKCEZ18}
Karapetyan, A.; Khonji, M.; Chau, C.; Elbassioni, K.~M.; and Zeineldin, H.~H. 2018.
\newblock Efficient Algorithm for Scalable Event-Based Demand Response Management in Microgrids.
\newblock \emph{{IEEE} Trans. Smart Grid}, 9(4): 2714--2725.

\bibitem[{Khonji et~al.(2019)Khonji, Karapetyan, Elbassioni, and Chau}]{KhonjiKEC19}
Khonji, M.; Karapetyan, A.; Elbassioni, K.~M.; and Chau, S.~C. 2019.
\newblock Complex-demand scheduling problem with application in smart grid.
\newblock \emph{Theor. Comput. Sci.}, 761: 34--50.

\bibitem[{Klimm et~al.(2021)Klimm, Pfetsch, Raber, and Skutella}]{KlimmPfetsch}
Klimm, M.; Pfetsch, M.; Raber, R.; and Skutella, M. 2021.
\newblock Approximation of Binary Second Order Cone Programs of Packing Type.

\bibitem[{Klimm et~al.(2022)Klimm, Pfetsch, Raber, and Skutella}]{KlimmPRS22}
Klimm, M.; Pfetsch, M.~E.; Raber, R.; and Skutella, M. 2022.
\newblock Packing under convex quadratic constraints.
\newblock \emph{Math. Program.}, 192(1): 361--386.

\bibitem[{Weymouth(1912)}]{W12}
Weymouth, T. 1912.
\newblock Problems in natural gas engineering.
\newblock \emph{Trans. Am. Soc. Mech. Eng.}, 34: 185--231.

\bibitem[{Wierman, Andrew, and Tang(2012)}]{tang12}
Wierman, A.; Andrew, L. L.~H.; and Tang, A. 2012.
\newblock Power-Aware Speed Scaling in Processor Sharing Systems: Optimality and Robustness.
\newblock \emph{Perform. Eval.}, 69(12): 601–622.

\bibitem[{Woeginger(2000)}]{Woeginger00}
Woeginger, G.~J. 2000.
\newblock When Does a Dynamic Programming Formulation Guarantee the Existence of a Fully Polynomial Time Approximation Scheme ({FPTAS})?
\newblock \emph{{INFORMS} J. Comput.}, 12(1): 57--74.

\bibitem[{Wood and Wollenberg(2012)}]{allen12}
Wood, A.~J.; and Wollenberg, B.~F. 2012.
\newblock \emph{Power generation, operation, and control}.
\newblock Wiley.

\bibitem[{Yu and Chau(2013)}]{YuC13}
Yu, L.; and Chau, C. 2013.
\newblock Complex-demand knapsack problems and incentives in {AC} power systems.
\newblock In Gini, M.~L.; Shehory, O.; Ito, T.; and Jonker, C.~M., eds., \emph{International conference on Autonomous Agents and Multi-Agent Systems, {AAMAS} '13, Saint Paul, MN, USA, May 6-10, 2013}, 973--980. {IFAAMAS}.

\end{thebibliography}

\end{document}